\documentclass[11pt]{article}
\usepackage{mathrsfs}
\usepackage{latexsym,lineno}
\usepackage{epsfig}
\usepackage{color}
\usepackage{amsmath}\usepackage{fleqn}\usepackage{verbatim}\usepackage{epsf}
\usepackage{amsthm}\usepackage{graphicx, float}\usepackage{graphicx}
\usepackage{amsfonts}\usepackage{amssymb}\usepackage{graphpap}
\usepackage{epic}\usepackage{curves}

\newcommand{\be}{\begin{equation}}
\newcommand{\ee}{\end{equation}}
\newcommand{\benum}{\begin{enumerate}}
\newcommand{\eenum}{\end{enumerate}}
\newcommand{\bit}{\begin{itemize}}
\newcommand{\eit}{\end{itemize}}
\newtheorem{thm}{Theorem}

\newtheorem{rem}{Remark}

\newtheorem{defn}{Definition}

\usepackage{graphicx}

\begin{document}
\def\s{\subseteq}
\def\n{\noindent}
\def\se{\setminus}
\def\dia{\diamondsuit}
\def\la{\langle}
\def\ra{\rangle}


\title{Structure Properties of Koch Networks
Based on Networks Dynamical Systems
}

\author{Yinhu Zhai$^a$, Jia-Bao Liu$^b$\footnote{  Corresponding author: Jia-Bao Liu, Shaohui Wang. E-mail addresses: Y. Zhai(zhaiyh@gdut.edu.cn),   J.B.  Liu (liujiabaoad@163.com),  S. Wang  ( shaohuiwang@yahoo.com) },  Shaohui Wang$^{c*}$\\
\small\emph {a. School of Information Engineering, Guangdong University of Technology, Guangzhou 510006, China}\\
\small\emph {b. Department of Mathematics, Anhui Jianzhu University, Hefei 23060, China}\\
\small\emph {c. Department of Mathematics and Computer Science, Adelphi University, Garden City, NY 11530, USA}
}
\date{}
\maketitle

\begin{abstract}
We introduce an informative labeling algorithm for the vertices of a family of Koch networks. Each of the labels is consisted of two parts, the precise position and the time adding to Koch networks. The shortest path routing between any two vertices is determined only on the basis of their labels, and the routing is calculated only by few computations. The rigorous solutions of betweenness centrality for every node and edge are also derived by the help of their labels. Furthermore, the community structure in Koch networks is studied by the current and voltage characteristics of its resistor networks.

\end{abstract}

\section{Introduction}
The WS small-world models \cite{1} and BA scale-free networks \cite{2} are two famous random networks which caused in-depth understanding of various physical mechanisms in empirical complex networks. The two main shortcomings are the uncertain creating mechanism and huge computation in analysis. Deterministic models always have important properties similar to random models, such as scale-free and small-world and high clustered, thus it could be used to imitating empirical networks appropriately. Hence the study of the deterministic models of complex network has increasing recently.

Inspired by simple recursive operation and techniques of plane filling and generating processes of fractal, several deterministic models \cite{3}-\cite{15}  have been created imaginatively and studied carefully. The famous Koch fractals \cite{16},  its lines are mapped into vertices, and there is an edge between two vertices if two lines are connected, then the generated novel networks was named Koch networks \cite{17}.  This novel class of networks incorporates some key properties which are characterized the majority of real-life networked systems: a power-law distribution with exponent in the range between 2 and 3, a high clustering coefficient, a small diameter and average path length and degree correlations. Besides, the exact numbers of spanning trees, spanning forests and connected spanning subgraphs in the networks is enumerated by Zhang et al in \cite{17}. All these features are obtained exactly according to the proposed generation algorithm of the networks considered \cite{001}-\cite{002}, \cite{20}-\cite{27}. 

However, some important properties in Koch networks, such as vertex labeling, the shortest path routing algorithm and length of shortest path between arbitrary two vertices, the betweenness centrality, and the current and voltage properties of Koch resistor networks have not yet been researched. In this paper, we introduced an informative labeling and routing algorithm for Koch networks. By the intrinsic advantages of the labels, we calculated the shortest path distances between arbitrary two vertices in a couple of computations. We derived the rigorous solution of betweenness centrality of every node and edge, and we also researched the current and voltage characteristics of Koch resistor networks.

\section{Koch networks}

The Koch networks are constructed in an iterative way. Let  $K_{m,t}$ denotes the Koch networks after $t \in  N $   iterations, and in which $N^*$  is a structural parameter.

\begin{defn} The Koch networks  $K_{m, t}$ are generated as follows: Initially $(t = 0)$,  $K_{m, 0}$ is a triangle. For  $t \geq 1$, $K_{m, t}$  is obtained from $K_{m, t-1}$  by adding $m$  groups of vertices to each of the three vertices of every existing triangles in  $K_{m,t-1}$. 
\end{defn}

\begin{rem}  Each group is consisted of two new vertices, be called son vertices. For both of the sons and their father vertex are connected to one another, the three vertices shaped a new triangle. 
\end{rem}

That is to say, we can get  $K_{m,t}$ from $K_{m,t-1}$  just by replacing each existing triangle in $K_{m,t-1}$ with the connected clusters on the right-hand side of Figure 1.

\begin{figure}[htbp]
   \centering
   \includegraphics[width=5in]{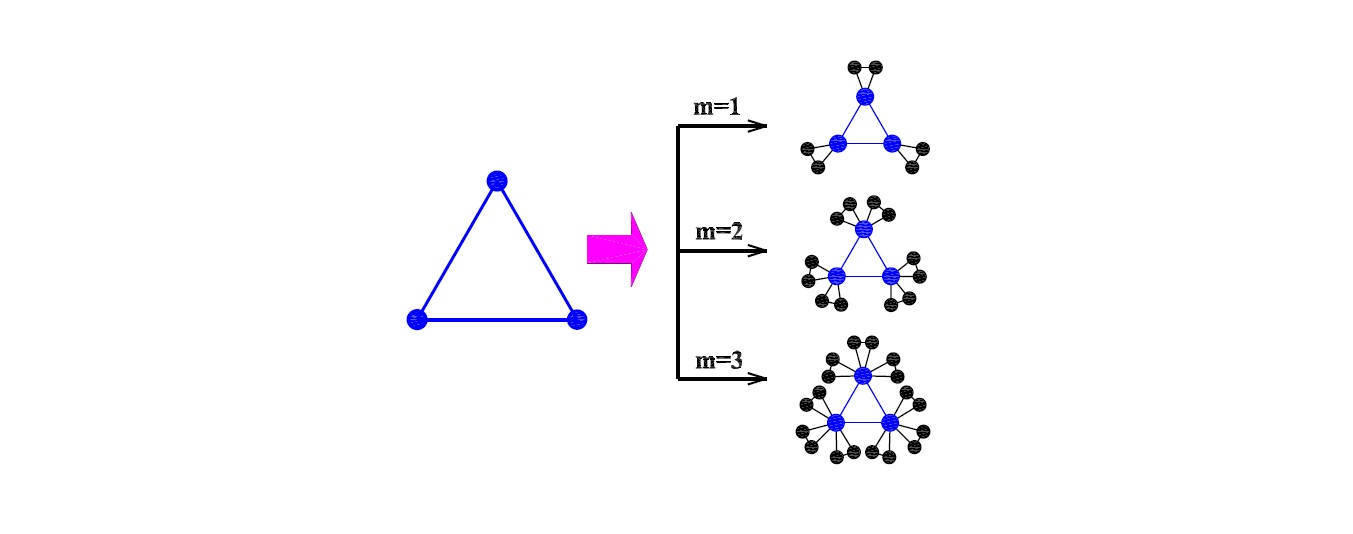}
   \caption{ Iterative construction method for the Koch networks when  $m = 1, 2, 3$. \cite{p1}}
 \label{fig: te}
\end{figure}

Some important properties of Koch networks are derived as below. The numbers of vertices and edges, i.e. order and size, in networks $K_{m,t}$ are 

\begin{equation} \label{eq:1}
N_t = 2(3m+1)^t +1,
\end{equation}
and 
\begin{equation}
E_t = 3(3m+1)^t.
\end{equation}

By denoting  $\Delta_v(t)$ as the numbers of nodes created at step  $t$, we obtained  $\Delta_v(t) = 6m(3m+1)^{t-1}$, then we also got that the degree distribution is  $P(k=2(m+1)^{t-i}) = 6m(3m+1)^{i-1} / [2(3m+1)^t +1]$, by substituting  
$i= t- ln(\frac{k}{2})/ln（m+1）$ in it, in the infinite $t$  limit, it gives

\begin{equation}
P(k) = \frac{3m}{3m+1} 2^{\frac{ln(3m+1)}{ln(m+1)}}k^{-\frac{ln(3m+1)}{ln(m+1)}}. 
\end{equation}

Then the exponent of degree distribution is $\gamma = ln(3m+1)/ln(m+1)$, which is belong to the interval $(1,2]$. The average clustering coefficient $C$ of the whole network is given by  $C = \frac{1}{N_t} \sum_{r=0}^{t} \frac{1}{k_i(r) - 1}L_u(r)$. When  $m$ is increased from $1$ to infinite, $C$  is increased from $0.82008$ to $1$. So, the Koch networks are highly clustered. The average path length (APL) approximates  $4mt/(3m+1)$ in the infinite  $t$, for APL is

\begin{equation}
d_t = \frac{3m+5+(24mt+24m+4)(3m+1^t)      }{3(3m+1)[2(3m+1)^t +1]     }: \frac{4mt}{3m+1}.
\end{equation}

It shows that Koch networks exhibit small-world behavior.
These properties indicated that Koch networks incorporate some key properties characterizing a majority of empirical networks: of simultaneously scale-free, small-world, and highly clustered.\cite{17}

\section{Vertex labeling algorithm}

\begin{defn}
All the vertices are located in three different sub-networks of Koch network, the label $n (n= 1, 2 $ or $ 3)$ is used to denote the sub-networks.
\end{defn}

\begin{rem}
Denote the three symmetrical sub-networks in Koch networks $K_{m,t}$ as  $K_{m,t}^1$, $K_{m,t}^2$  and  $K_{m,t}^3$, then $K_{m,t}$ is obtained just by linking the hub of three sub-networks directly. Therefore, the label $n (n = 1,2 or 3)$ is used to distinct the vertices in the three different sub-networks $K_{m,t}^n$.
\end{rem}

A binary digits code is used to identify the precise position of a vertex in $K_{m,t}^n$  and the exact time which is linked to  $K_{m,t}^n$, the method is shown as below.

\begin{defn}
Any vertex in  $K_{m,t}^n$ is marked with binary digits $b_1b_2b_3 \dots b_j$, where $b_1= 0$ when $j=1.$  $b_j = 0 ~ or~ 1$ when $j= 2,3,,\dots, t$. The $0~(or ~1)$ in binary digits represented that the new $m$  vertices are grown from a son vertex (or father vertex) in a triangles. The length of the binary digits is the time of the vertex which is linked into Koch networks.
\end{defn}

\begin{rem}
Because the initial network $K_{m,0}$  is a triangle, all the three initial vertices in it have no father vertices, so that the new vertices adding to the initial vertices should marked with $0$ at time  $j=1$, that is, $b_1$  must be $0$. 
\end{rem}

Then, we obtained the set  $S$, possessing all the binary digits codes of each vertices in  $K_{m,t}$, as below

\begin{equation}
S = \{\Phi, b_1, b_1b_2, b_1b_2b_3, \dots, b_1b_2b_3\dots b_t \}.
\end{equation}

\begin{rem}
The element $\Phi$ in $S$ implies that, when  $t=0$, the length of $\Phi$  is zero in  $K_{m,0}$. 
\end{rem}

The Definition 3  ensures that all the vertices, adding to an existing vertex at step  $j$, have the same binary codes  $b_1b_2b_3 \dots b_j$. Consequently, the number of vertices which are added to an existing father vertex at step $j$  is given by

\begin{equation}
l_{max} (j) = (2m) ^{j- \sum_{i=1}^{j} b_i} (m+1)^{\sum_{i=1}^{j} b_i}.
\end{equation}

So that we need to mark the vertices of this group with an extra integer  $l(j) \in [1, l_{max}(j)]$ for they all have the same binary codes $b_1b_2b_3 \dots b_j$ and the same group indicator  $n$.

\begin{defn}
An integer $l(j)$ is used to identify the precise position, increasing by clockwise direction, of a vertex in the group which are added to a father vertex at the iteration  $j$.
\end{defn}

\begin{rem}
Because  $l(j)$ is increased from $1$ and is positioned after the binary codes, a dot is needed to insert into the integer $l(j)$  and the binary codes for avoiding confusion.
\end{rem}

In sum, arbitrary vertex which is added to $K_{m,t}$ at step  $j$ will label with  $nb_1b_2b_3 \dots b_j l(j)$. The code $n$ denotes which sub-networks of $K_{m,t}^n$  is the vertex belonging to; the binary digits  $b_1b_2b_3 \dots b_j$ indicates which father vertex it is linking to; the positive integer  , which is increasing  by clockwise, is used in marking the precise position around a father vertex.

Define the set $M(j)$  as the label set of the vertices which are adding to networks  $K_{m,t}$ at step  $j$, it is apparently that   $M(0) = \{1, 2, 3\}$ and  $M(j) = \{nb_1b_2b_3 \dots b_j l(j)\}$.  Let the set $L_{m,t}$  represents all the label of all vertices in  $K_{m,t}$, we obtained

\begin{equation}
L_{m,t} = \bigcup_{j=1}^t M(j).
\end{equation}

For example, Figure 2 demonstrates the vertex labelling of all the vertices in Koch network  $K_{2,2}$. In the following sections, we deduced some important properties of Koch networks just on the basis of the labels of their vertices.

\begin{figure}[htbp]
   \centering
   \includegraphics[width=4in]{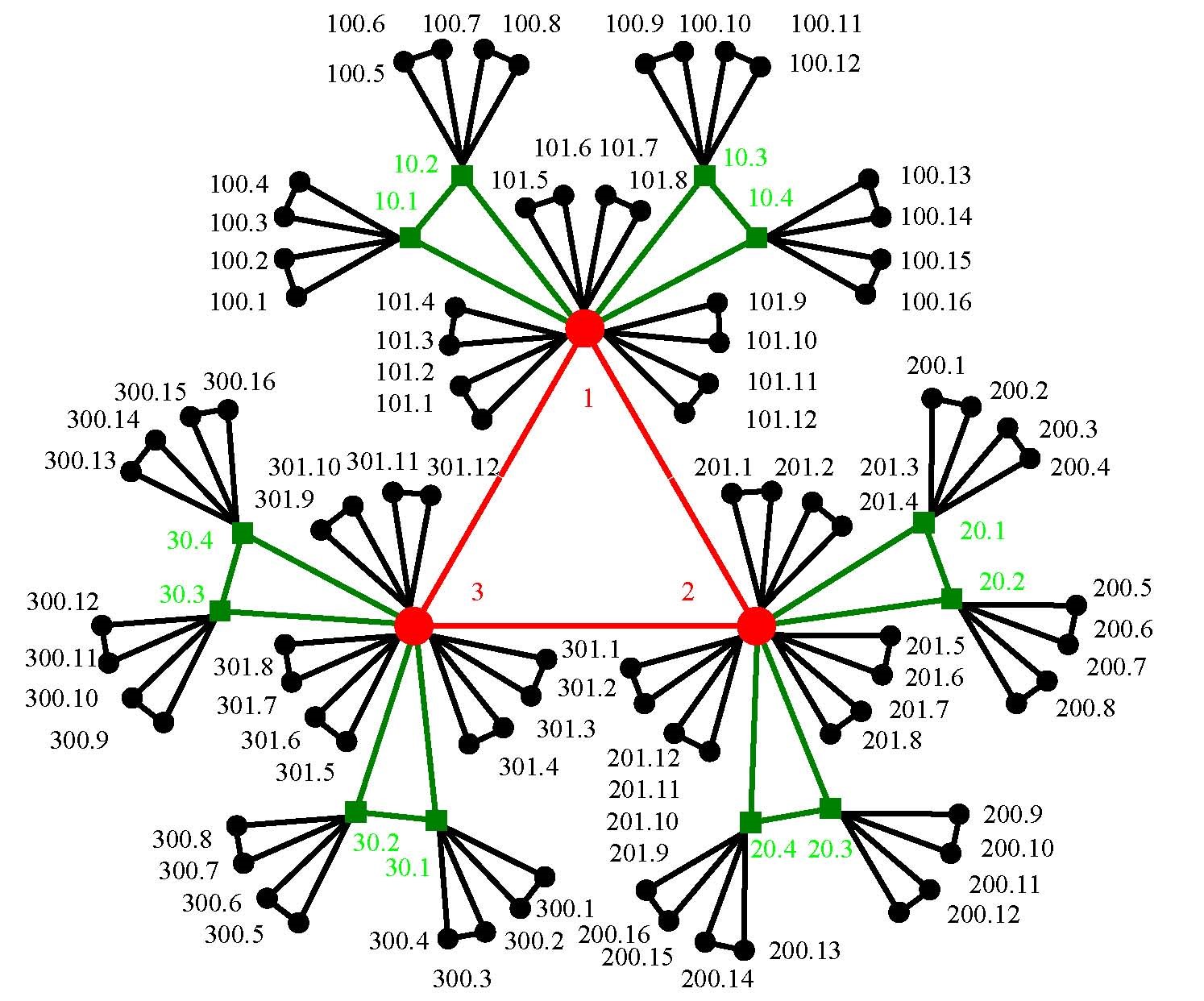}
   \caption{The labeling of Koch network $K_{2,2}$ when $m=2$ and  $t=0, 1, 2$. The red dots are initial vertices; the green squares are vertices adding at  $t=1$; the black dots denote vertices adding to network at  $t=2$.}
 \label{fig: p2}
\end{figure}

\begin{thm}
Each vertex has a unique label.
\end{thm}

\begin{proof}
Suppose an arbitrary vertex labels with $nb_1b_2b_3 ... b_j l(j)$. Firstly, from the labeling algorithm, the labels of any pair vertices are different from each other. Secondly, the size of $L_{m,t}$  equals the size of Koch networks. So, we deduced that any vertex has a unique label.
    
Assume that  $nb_1b_2b_3 ... b_i l_v(i)$ is the label of arbitrary vertex $v$  which is adding to $K_{m,t}$  at step  $i$, and let the set $A(v)$  denotes the labels of all neighbor vertices of  $v$. By comparing the vertex's degree between   $v$ and its neighbors, $A(v)$  can be divided into three subsets:  $A_e(v)$, $A_l(v)$   and  $A_h(v)$, the vertices in which sets have degree equals, lower and higher than the degree of  $v$, respectively. That is to say,  $A(v) = A_e(v) \cup A_l(v) \cup A_h(v)$. 

\end{proof}

\begin{thm}
$A_e(v) = \{nb_1b_2b_3 \dots b_i l_e(i)\}$, where vertex degree $l_e(i) = l_v(i) +1$ if $mod(l_v(i), 2) = 1$, or $l_e(i) = l_v(i) -1$ if $mod(l_v(i), 2) = 0$.
\end{thm}

\begin{proof}
From the construction algorithm of  $K_{m,t}$, any father vertex will add $m$  group vertices at each step, and every group vertices is consisted of two vertices, then three of them is linked to each other and formed a new triangle. Therefore, the two vertices in the same group are neighbors which are linking directly and have the same degrees. By the labeling method, the $m$  group vertices labels with the integers $l(i)$  which increasing from $1$ to $l_{max}(i)$  by clockwise. So that,  $nb_1b_2b_3 \dots b_i (l_v(i) +1)$  is the neighbor of   $nb_1b_2b_3 \dots b_i l_v(i) $  with the same degrees if $mod (l_v(i), 2) = 1$, or $nb_1b_2b_3 \dots b_i (l_v(i) -1)$  is the neighbor if $mod (l_v(i), 2) = 0$. 
\end{proof}

\begin{thm}
$A_l(v) = \{nb_1b_2b_3 \dots b_i 0 l_v(i+1), nb_1b_2b_3 \dots b_i 0 l_v(i+2), \dots, nb_1b_2b_3 \dots b_i 0 l_v(t)\}.$
\end{thm}

\begin{proof}
From the labeling algorithm, the vertices with longer binary codes have lower degrees than the vertices with shorter binary codes. In addition, the $0$ or $1$ in binary codes indicates the new vertex is growing from the two son vertices or father vertex in each triangle. Hence we can understand that the vertices, adding to $v$  at steps  $i+1$,  $i+2$, ... ,  t, is labeled with $nb_1b_2b_3 \dots b_i 0 l_v(i+1), nb_1b_2b_3 \dots b_i 0 l_v(i+2), \dots, nb_1b_2b_3 \dots b_i 0 l_v(t)$. 
\end{proof}

Define $\lceil{x} \rceil = ceil(x)$ as   the function returning the biggest integer just smaller than real number $x$.

\begin{thm}
$A_h(v) = \{nb_1b_2b_3 \dots b_{j-1}. \lceil {l_v(i) / 2m(m+1)^{\sum_{k=j+1}^{i} b_k}} \rceil\}$.

\end{thm}

\begin{proof}
Knowing that $nb_1b_2b_3 \dots b_{i}l_v(i)$  is the label of an arbitrary vertex  $v$. From the construction mechanism, we obtained that the label of the only father vertex of $v$  is depending on the composition in binary codes $b_1b_2b_3 \dots b_i$  of vertex  $v$. Suppose the first $0$ in  $b_1b_2b_3 \dots b_i$, from right to left side is  $b_j$. By the construction method, it is clearly that $v$  is linked to a vertex with higher degree which is labeled with  $nb_1b_2b_3 \dots b_{j-1}. \lceil {l_v(i) / 2m(m+1)^{\sum_{k=j+1}^{i} b_k}} \rceil$. In particular, if the first $0$ of $b_1b_2b_3 \dots b_i$  is  $b_1$, the vertex with higher degree is exactly a hub of Koch networks which is labeled with  $n=1, 2~or ~3$.  
\end{proof}

\section{Routing by Shortest Path}

The deterministic models of complex network always have fixed shortest path, but how to mark it only by their labels is rarely researched\cite{15}. The following rules are used to determine the shortest path routing between any two vertices by the help of their labels. Let $nb_1b_2b_3 \dots b_i l $  and  $n'b_1'b_2'b_3' \dots b_j' l '$ as the labels of arbitrary pair of vertices in  $K_{m,t}$.

\begin{thm}
The shortest path routing algorithm in Koch networks.

If  $n \neq n'$, find out, by Theorem 4, all their higher degree neighbors of the two vertices, till the hubs  $n$ and  $n'$; then the shortest path is linked all vertices of them;

If  $n = n'$, the first step is marking higher degree neighbors till the common highest degree vertex by Theorem 4; then, judge the two second highest degree vertices are neighbors or not by Theorem 3; if not, the shortest path is connected all higher degree neighbors till the highest degree vertex; if yes, the shortest path is just the same as above but to eliminate the highest degree vertex.

\end{thm}

\begin{proof}
If  $n \neq n'$, the two vertices are located in different sub-networks $K_{m,n}^n$  and  $K_{m,n}^{n'}$. The routing by shortest path between two vertices in different subnets is ascertained as below. First, we obtained the neighbors which have higher degrees recursively by Theorem 4, till the hubs $n$  and  $n'$. Then, connect all of them in turn; it’s the only shortest path between two vertices. 

If  $n= n'$, it is clear that the shortest path is located in the same sub-networks  $K_{m,t}^n$. We found out the neighbors with higher degree by using Theorem 4 repeatedly, till the common highest degree vertex. Then, judge the two second highest degree vertices are neighbors or not by Theorem 3. If they are not neighbors, we determined the shortest path as above by linking all the higher degree vertices till the highest vertex, by the help of the construction method of Koch networks. Else if they are neighbors, the shortest path is as the same as above by excluding the highest degree vertex.   

\end{proof}

The shortest path between any pair vertices in $K_{m,t}$  is obtained after no more than  $2t$ times of ceil computations and modulo operations by the help of labeling method and routing algorithm proposed in this research. That is to say, the shortest path routing and the shortest distance between arbitrary pair of vertices in Koch networks can be dealt out in few computations.

\section{Betweenness Centrality}

Betweenness centrality is originated from the analysis of the importance of the individual in social networks, including the betweenness of any vertex and edge in networks. If the betweenness of a node/edge is bigger, then the node/edge is in the social network is more important. \cite{2} The betweenness of a vertex   for undirected networks is given by the expression
\begin{equation}
g(v) = \frac{\sum_{s \neq v \neq t} \sigma_{st}(v)}{(N_t - 1)(N_t -2) / 2},
\end{equation}
where $\sigma_{st}(v)$ is the number of the shortest paths which are passing through  $v$. The computation of betweenness is very difficult in most networks. Fortunately, the betweenness of Koch networks can be derived qualitatively and quantitatively by the help of their labels in Koch networks, which is shown as below.

Suppose that an arbitrary vertex  $v$, which is adding to $K_{m,t}$  at time $i$, is labeled with  $nb_1b_2b_3 \dots b_i l_v(i)$. The vertices in $K_{m,t}$  can be divided into three parts: the vertex  $v$, the offspring vertices which are connected to $v$  directly and indirectly after step  $i$ (they all have lower degrees than $v$), the third part is the other vertices in  $K_{m,t}$.  Assume that the number of the second part vertices is  $N_l$, and it can be worked out that  $N_l = [2(3m+1)^{t-i} -2]/3$ by equations (1). Apparently the number of the third part is  $N_t - N_l - 1$. For the shortest path routing between any two vertices is unique, we got that  $\sum_{s \neq v \neq l} \sigma_{st} (v) = N_t(N_t - N_l - 1)/ 2$. Substitute this equation and equation (1) into equation (8), the betweenness of a vertex which is labeling with $nb_1b_2b_3 \dotsb_i l_v(i)$  is given by

\begin{equation}
g(v) = \frac{2[(3m+1)^{t-1} - 1][3(3m+1)^t - (3m+1)^{t-i} - 1]}{ 3(3m+1)^t [2(3m+1)^t - 1]}.
\end{equation}

For $i = t- ln(k/2) / ln(m+1)$ and $a^{-ln(b)/ln(c)} = b^{-ln(a)/ln(c)}$, then we obtained that the formula $g(v)\sim c_1 k^{ln(3m+1) / ln(m+1)}$  which holds with  $c_1 > 0$. Therefore, the vertex betweenness in Koch networks is in exponentially proportional to the vertex's degree with an exponent $\gamma = ln(3m+1)/ ln(m+1)$  belonging to the interval $(1, 2].$ 

The betweenness of edges can also be deduced by similar way. Note $e$ as the edge between any two neighbor vertices $v$  and $u$  which are labeling with  $nb_1b_2b_3 \dots b_i l_v(i)$  and  $n'b_1'b_2' \dots b_j' l_u'(j)$. Without loss of generality, assume that vertex $u$  has higher degree than $v$. So that the label of $u$  belongs to the set  $A_h(v)$  by Theorem 4. Suppose a triangle are shaped by three vertices:  $v$,  $u$ and  $w$. Therefore,  $w$ has the degree same as  $v$. Then, Koch network $K_{m,t}$  can be divided into three parts: the lower degree vertices linking to  $v$  directly or indirectly, the vertices connected to  $u$ directly or indirectly, the lower degree vertices adding to   directly or indirectly, respectively. Correspondingly, the label set  $L_{m,t}$ will falls into three subsets:  $A_{el}(v)$,  $A_{other}(u)$ and  $A_{el}(w)$. The relationship of these four label sets is shown as below

\begin{equation}
L_{m,t} =  A_{el}(v) \cup A_{other}(u) \cup A_{el} (w). 
\end{equation}

The sizes of  $A_{el}(v)$, $A_{other}(u)$  and  $A_{el}(w)$ are derived as $N_{el}(v) = [2(3m+1)^{t-1} +1] / 3$, $N_{other}(w) = N_t - N_{el}(v) - N_{el}(w)$      and  $N_{el}(w) = N_{el}(v)$. For the shortest path between any two vertices is unique, then the betweenness of the edge  $e$ is defined as below

\begin{equation}
g(e) = \frac{N_{el}N_{other}/ 2}{ (N_t - 1) (N_t - 2) / 2}.
\end{equation}

Therefore, the betweenness centrality of the edge $e$  is given by

\begin{equation}
g(e) = \frac{[2(3m+1)^{t-i} + 1][ 6(3m+1)^t - 4(3m+1)^{t-i} + 1]}{ 18 (3m+1)^t [2(3m+1)^t - 1]}.
\end{equation}

Therefore, the edge betweenness holds  $g(e) \sim c_2 k ^{ln(3m+1) / ln(m+1)}$, where  $c_2 > 0$. The edge betweenness is also in exponentially proportional to the degree of the lower degree vertex  $v$, the exponent is    $\gamma = ln(3m+1) / ln(m+1)$ belonging to the interval $(1, 2]$. In a word, the betweenness of an edge is in exponentially proportional to the time of which is adding to Koch networks.

\section{Resistor networks}

The communities in networks are the groups of vertices within which the connections are dense, but between which the connections are sparser. A community detection algorithm which is based on voltage differences in resistor networks is described in \cite{18} and \cite{19}. The electrical circuit is formed by placing a unit resistor on each edge of the network and then applying a unit potential difference (voltage) between two vertices chosen arbitrarily. If the network is divided strongly into two communities and the vertices in question happen to fall in different communities, then the spectrum of voltages on the rest of the vertices should show a large gap corresponding to the border between the communities.

Moreover, the information in complex networks is not only always flow in the shortest paths; so that the evaluation of betweenness of nodes can also have the other principles, such as the current-flow betweenness. Consider an electrical circuit created by placing a unit resistor on every edge of the network. One unit of current is injected into the network at a source vertex and one unit extracted at a target vertex, so that the current in the network as a whole is conserved. Then, the current-flow betweenness of a vertex  is defined as the amount of current that flows through   in this setup, the average of the current flow over all source-target pairs is shown as below

\begin{equation}
g_v = \frac{\sum_{s < t} I_{st} (v)}{ N_t(N_t -1) / 2},
\end{equation}

where $I_{st}(v)$  is the current over  .

After placed a unit resistor on every edge in  $K_{m,t}$, then insert one unit of current or voltage at source vertex $v_0$  labeling with  $nb_1b_2b_3 \dots b_il$, further choose the target vertex $v_{m+1}$  with labels  $n'b_1'b_2' \dots b_j' l'$. Assume the shortest path is from $v_0$  to vertices  $v_1,..., v_m$  till  $v_{m+1}$. Therefore, the shortest distance is  $m+1$. The property of Koch resister networks is described as below.

\begin{thm}
If  $n \neq n'$, from Theorem 5, there are two hubs $v_j = n$  and  $v_{j+1} = n'$ with highest degree in shortest path. Hence, the vertices which are affected by unit voltage are $\{v_i \} \cup \{ \overline{v_k}\}$, where  $i= 0, 1, 2, \dots, m+1$, $k=0, 1, 2, \dots, j-1, j+1, \dots, m+1$, $\overline{v_i}$  is the neighbor of $v_i$  with the same degree, and apparently $\overline{v_{j+1}}$  is the other hub. The edges between vertices $\{ v_i \} \cup \{ \overline{v_k} \}$ formed $m+1$  triangles which are in series and the common vertices are $\{ v_i \}$, so that the unit current will only passed though these edges in whole Koch networks  $K_{m,t}$. 

If $n = n'$  and there are two highest degree vertices, noting $v_j$  and  $v_{j+1}$, in the shortest path, hence the unit voltage can only affected vertices $\{v_i \} \cup  \{ \overline{v_k} \}$ in  $K_{m,t}$, where  $i = 0, 1, 2, \dots, m+1$,  $k = 0, 1, 2, \dots, j-1, j+1, \dots, m+1$, $\overline{v_i}$  is the neighbor of   with the same degree too, but $\overline{v{j+1}}$  is a higher degree neighbor which is linked with $v_j$  and  $v_{j+1}$  directly; the unit current also flows the edges in  $m+1$ triangles which are in series.

If  $n = n'$, but there is the only highest degree vertices, denoting  $v_j$, in shortest path, the unit voltage impacts vertices $\{v_i \} \cup  \{ \overline{v_k} \}$, where $i = 0, 1, 2, \dots, m+1$,  $k = 0, 1, 2, \dots, j-1, j+1, \dots, m+1$, $\overline{v_i}$  is the neighbor of $v_i$  with the same degree; the behavior of unit current is same as the two conditions above.
\end{thm}

\begin{thm}

The voltages of vertices $\{v_i \}$ shape an arithmetic progression from $1$ to $0$, and the step length is  $\frac{1}{m+1}$. The voltage of vertices $\{ \overline{v_k} \}$ decrease from $1- \frac{1}{2m+2}$ to  $\frac{1}{2m+2}$, the step length is also $\frac{1}{m+1}$. 

\end{thm}

\begin{proof}
The proof of above is obvious by the help Theorem 6.  
\end{proof}

\begin{thm}

The current stream from the edges which are linked to the vertices $\{v_i \}$ is  $\frac{2}{3}$, while the current pass though the edges linking to $\{ \overline{v_k}\}$ is the remaining  $\frac{1}{3}$.

\end{thm}

\begin{proof}
The theorem can be proved easily by the help of the forming mechanism of Koch resistor networks, Theorem 5 and Theorem  6. 
\end{proof}

In brief, the spectrum of voltages on the vertices shows that Koch networks have no significant community structure in spite of having massive triangles between nodes. Also, the current-flow can gauge well the importance of edges betweenness in Koch networks in information flowing which is not flowing only by the shortest path.

\section{Conclusions}

The family of Koch networks, with properties of high clustering coefficient, scale-free, small diameter and average path length and small-world, successfully reproduces some remarkable characteristics in many nature and man-made networks, and has special advantages in the research of some physical mechanisms such as random walk in complex networks. 

We provided an informative vertex labeling method and produced a routing algorithm for Koch networks. The labels include fully information about any vertex’s precise position and the time adding to the networks. By the help of labels, we marked the shortest path routing and the shortest distance between any pair of vertices in Koch networks, the needed computation is just no more than $2t$  times of ceil computations and modulo operations. Moreover, we derived the rigorous solution of betweenness centrality of every vertex and edge in Koch networks, and we also researched the current and voltage characteristics in it on the basis of their labels. 

By the help of our results, in contrast with more usually probabilistic approaches, the deterministic Koch models will have unique virtues in understanding the underlying mechanisms between dynamical processes (random walk, consensus, stabilization, synchronization, and so on) to the structure of complex networks by  the new method of rigorous derivation.

\noindent{\bf Conflict of Interests}\\
The authors declare that there is no conflict of interests regarding the publication of this paper.

\noindent{\bf Acknowledgments}\\
The work of was supported by National Science Foundation of China under Grant Nos. 61273219, 11471016, 11601006
  and 11401004.



%
%

\end{document}